\documentclass{article}
\usepackage[english]{babel}
\usepackage{graphicx} % Required for inserting images
\usepackage{algorithm}
\usepackage[noend]{algpseudocode}
\usepackage{amsmath}
\usepackage{amsthm}
\usepackage{amssymb}
\usepackage{mathalfa}
\usepackage[letterpaper, portrait, margin=1.2in]{geometry}
\usepackage{marginnote}
\usepackage{thmtools}
\usepackage{hyperref}
\usepackage{cleveref}
\usepackage{relsize}
\usepackage{bbm}
\usepackage{multirow}
\usepackage{float}
\usepackage{xargs}    
\usepackage[pdftex,dvipsnames]{xcolor}  % Coloured text etc.
\usepackage{relsize}
\usepackage[normalem]{ulem}
\usepackage{pgfplots}
\usepgfplotslibrary{fillbetween}
\pgfplotsset{compat=1.18}
\usepackage{dsfont}

\makeatletter
\newcommand{\AlignFootnote}[1]{%
	\ifmeasuring@
	\else
	\iffirstchoice@
	\footnote{#1}%
	\fi
	\fi}
\makeatother

\declaretheorem[name=Theorem, numberwithin=section]{theorem}

\declaretheorem[name=Lemma, sibling=theorem]{lemma}
\declaretheorem[name=Claim, sibling=theorem]{claim}
\declaretheorem[name=Fact, sibling=theorem]{fact}

\declaretheorem[name=Remark, numbered=no]{remark}

\newcommand{\Cov}[0]{\textup{Cov}}
\newcommand{\floor}[1]{\lfloor #1 \rfloor}
\newcommand{\pnorm}[2]{\lVert #1 \rVert_{#2}}

\newcommand{\E}{\mathop{\mathbb{E}}}

\title{A $4.509$-Approximation Algorithm for Generalized Min Sum Set Cover}

\author{Amey Bhangale\thanks{Department of Computer Science and Engineering, University of California, Riverside. Supported by the Hellman Fellowship award and NSF CAREER award 2440882.}
    \and
Yezhou Zhang\thanks{Department of Computer Science and Engineering, University of California, Riverside. Supported by the Hellman Fellowship award and NSF CAREER award 2440882.}}

\date{}

\begin{document}

\maketitle

\begin{abstract}
We study the \emph{generalized min-sum set cover} (GMSSC) problem, where given a collection of hyperedges $E$ with arbitrary covering requirements $\{k_e \in \mathbb{Z}^+ : e \in E\}$, the objective is to find an ordering of the vertices that minimizes the total cover time of the hyperedges. A hyperedge $e$ is considered covered at the first time when $k_e$ of its vertices appear in the ordering.

We present a $4.509$-approximation algorithm for GMSSC, improving upon the previous best-known guarantee of $4.642$~\cite[SODA'21]{BansalBFT21}. Our approach retains the general LP-based framework of Bansal, Batra, Farhadi,
and Tetali~\cite{BansalBFT21} but provides an improved analysis that narrows the gap toward the lower bound of $4$-approximation assuming P$\neq$NP. Our analysis takes advantage of the constraints of the linear program in a nontrivial way, along with new lower-tail bounds for the sums of independent Bernoulli random variables, which could be of independent interest.
\end{abstract}

\section{Introduction}

In this paper, we study approximation algorithms for the following related problems.
\subparagraph*{Min Sum Set Cover Problem:} The min sum set cover problem (MSSC) takes a universe of $n$ elements, denoted as $[n]$, as input and a collection of it subsets. Given a permutation of all elements in $[n]$, a subset $S$ is considered covered at the earliest time an element from $S$ appears in the permutation. The goal of the min sum set cover problem is to find a permutation of $[n]$ that minimizes the sum of the cover time of all given subsets. Alternatively, this problem can be described on a hypergraph: given a hypergraph $G=(V,H)$ with $n$ vertices and a set of hyperedges $e\in H$. For a schedule (permutation) of $n$ vertices, the cover time is defined as the earliest time that one of $v\in e$ appears in the schedule. The goal is to find a schedule that minimizes the sum of the cover time of all hyperedges. This problem was formally defined in~\cite{FeigeLT02}, which also introduced a $4$-approximation greedy algorithm and a matching hardness result.

\subparagraph*{Generalized Min Sum Set Cover Problem:} The generalized min sum set cover problem (GMSSC), also known as the multiple intents re-ranking problem generalized MSSC by assigning each hyperedge an individual cover requirement $k_e\leq |e|$, such that the hyperedge is considered at the earliest time when $k_e$ of its vertices appear in the schedule. The problem was first introduced by Azar, Gamzu, and Yin in \cite{AzarGY09}, which also provided an $O(\log(\max{k_e}))$-approximation algorithm for the problem. Bansal, Gupta, and Krishnaswamy first provided a constant, $485$-approximation in \cite{BansalGK10} with a linear program relaxation, which included strengthening constraints known as KC (Knapsack Cover) inequalities from \cite{CarrFLP00}.  Skutella and Williamson later improved the approximation factor to $28$ in \cite{SkutellaW11} by applying the $\alpha$-point rounding algorithm to the LP relaxation, and Im, Sviridenko, and Zwaan further improved the factor to $12.4$ in \cite{ImSZ14}. Recently, Bansal, Batra, Farhadi, and Tetali gave a $4.642$-approximation in \cite{BansalBFT21} by adding a linear transformation, or the kernel, to the LP fractional optimal solution and rounding based on the adjusted solution. This remains the best-known approximation to date. The hardness of GMSSC is $4-\epsilon$, derived from MSSC in~\cite{FeigeLT02}.

\subparagraph*{Min Latency Set Cover Problem:} The min latency set cover problem, first formally studied in \cite{HassinL05}, is a special case of GMSSC where each hyperedge has a cover requirement $k_e=|e|$. It is a well-studied problem, also known as the special case of the classic precedence-constrained scheduling problem, and has many $2$-approximation algorithms (see, e.g, \cite{KargerSW10} for a survey). Bansal and Khot proved the Unique-Game hardness of this problem is $2-\epsilon$ in \cite{BansalK09}.

\subsection{Our Results}
Our main theorem improves the $4.624$-approximation for GMSSC given by Bansal, Batra, Farhadi, and Tetali in~\cite{BansalBFT21}.

\begin{restatable}{theorem}{main}
\label{thm:main}
    There exists a $4.509$-approximation algorithm for the generalized min sum set cover (GMSSC) problem. 
\end{restatable}

In~\cite{BansalBFT21}, the authors gave an approximation algorithm for GMSSC and showed that the algorithm achieves an approximation ratio of $4.624$. They also remarked that the numerical analysis suggests that the approximation factor is no more than $4.5232$, even though they were unable to prove this analytically. Finally, they showed that getting better than a $4.509$ factor approximation has a natural bottleneck for their approach. Towards proving the main theorem above, we provide a tighter analysis of their algorithm and show that it indeed achieves a $4.509$ approximation ratio. Our analysis leverages the linear program’s constraints in a non-trivial manner along with new lower-tail bounds for the sums of independent Bernoulli random variables, which could be of independent interest (see Section~\ref{sec:techniques} for more details).

Our second theorem gives a $2$-approximation algorithm for the min latency set cover problem using the framework of Bansal, Batra, Farhadi, and Tetali \cite{BansalBFT21}, improving upon the $e = 2.718\ldots$ approximation algorithm given in the same paper. 
\begin{restatable}{theorem}{second}
\label{thm:main2}
    There exists a $2$-approximation algorithm for the min latency set cover problem. 
\end{restatable}

Recall that a $2$-approximation algorithms were already known before~\cite{KargerSW10}. We designed a new kernel which, when plugged into the framework of~\cite{BansalBFT21}, gives the optimal approximation algorithm for the min latency set cover problem.

In ~\cite{BansalBFT21}, the authors already showed that their framework gave an optimal $4$-approximation algorithm for MSSC (all $k_e = 1$). Theorem~\ref{thm:main2} now shows that the framework provides an optimal approximation algorithm (assuming the Unique Games Conjecture) for min sum latency set cover (all $k_e = |e|$). It would be interesting to see if it gives an optimal approximation algorithm for GMSSC (arbitrary $k_e$). 

\subsection{Techniques}
\label{sec:techniques}
This section provides an overview of the LP-rounding based algorithmic techniques employed in this paper.

The foundation of our approach lies in the work of \cite{BansalGK10}, which defined a commonly used linear program (LP) for MSSC and GMSSC. This LP returns a fractional solution $x$ where $x_{v,t}$ is an indicator of whether vertex $v$ is scheduled at time $t$. The core LP constraint is given by
\[\sum_{v\in e}\sum_{t' < t}x_{v,<t}\geq \sum_{t'< t}x_{e,t'},\quad\forall e,t,\]
where $\sum_{t'<t}x_{e,t'}$ indicates the fraction of edge $e$ covered at time $t$. However, a simple randomized rounding algorithm, rounding based on $\sum_{t'\leq t}x_{v,t'}$, proves ineffective even for MSSC. This is because if the LP fractionally covers a hyperedge of $n$ vertices by scheduling each vertex to the extent of $1/n$ at time $t$, this approach leaves this hyperedge uncovered with probability $(1-1/n)^n\approx 1/e\gg 0$ for any time after $t$. For a long period after $t$, the LP solver may consider this hyperedge covered, but yet no vertex might be scheduled in an integral solution, leading to an arbitrarily large approximation ratio. To overcome this, a linear transform, or kernel $K$, was introduced and applied to the fractional optimal solution $x$ from the LP in \cite{BansalBFT21}.
% \abnote{it would be useful to state the LP variables and their intuitive meaning}. 
The standard $\alpha$-point randomized rounding was then applied to $Z=Kx$ to obtain a feasible integral solution. Using a kernel $K(t,t')=\beta/t~(t'\leq t)$, for $\beta=2$, in \cite{BansalBFT21}, the authors achieved a $4$-approximation for MSSC, which corresponded to the known hardness result. This was accomplished by analyzing the expected cover time of a hyperedge in the fractional optimal solution $x$, denoted $c_x(e)$, and its corresponding expected time in the integral solution $z$, denoted $c_z(e)$, along with a random tie-breaking rule for the final integral solution. While $c_x(e)$ was obtained from the optimal fractional solution, $c_z(e)$ was related to the sum of random Bernoulli variables. The approximation factor was then formulated as a convex optimization problem involving the ratio $c_z(e)/c_x(e)$. Through convex optimization techniques, it was proven that the worst-case scenario occurs when the LP solver assigns the extent $1/n$ to all vertices. Nevertheless, the $\beta/t$ kernel effectively increases the probability of scheduling $v\in e$ over time, leading to a feasible solution.

For the GMSSC problem, the best known hardness is $4$, and all previous approximations follow the extended analysis of MSSC. In~\cite{BansalGK10}, an unbounded integrality gap is demonstrated for the natural extended LP for GMSSC, and a set of strengthening constraints, KC (Knapsack Cover) inequalities from \cite{CarrFLP00}, was introduced for GMSSC.
\[\sum_{v\in e\setminus S}\sum_{t'<t}x_{v,t'}\geq (k_e-|S|)\sum_{t'< t}x_{e,t'},\quad \forall e,t,S\subseteq e.\]
The KC inequalities eliminated the unbounded integrality gap and enabled the proof of several constant-factor approximation algorithms. However, applying the same algorithm and analysis to GMSSC in \cite{BansalBFT21} yielded only a $4.642$-approximation factor. 
This result left a gap between GMSSC and MSSC, and the method to eliminate this gap remained unknown. For instance, if $k_e-1$ vertices are already scheduled in the integral solution (let $S$ be the set of these vertices), the above KC inequality effectively transforms into a constraint in the LP for MSSC. This suggests that GMSSC could potentially be reduced to MSSC after a certain time.

Our paper aims primarily to eliminate this gap between GMSSC and MSSC after a certain time in the scheduling. To this end, we address a key analytical challenge in GMSSC: \textit{for a hyperedge $e$, the remaining cover requirement at time $t$, $k_e(t)$, potentially decreases over time as vertices are getting scheduled}. This dynamic nature makes the probability of covering $e$ difficult to accurately express with a single function. We establish a novel upper bound $P^{k_e(t)}$ on the probability that a hyperedge $e$ remains uncovered at time $t$, explicitly involving $k_e(t)$. Furthermore, by leveraging multiple KC inequalities, we demonstrate that the core KC inequality (after applying the kernel) $\sum_{v\in e\setminus A_e(t)}z_{v,<t}\geq k_e(t)z_{e,<t}$ can be more accurately expressed as:
\[\sum_{v\in e\setminus A_e(t)}z_{v,<t}\approx k_e(t)z_{e,<t}+\delta(t),\]
where $A_e(t)$ is the set of vertices guaranteed to be scheduled in the integral solution before time $t$, $\delta(t)$ is a crucial gap extracted from other KC inequalities.

Consequently, we derive a more precise expression for $c_z(e)$, the expected cover time in the integral solution. More precisely, since a vertex $v\in e\setminus A_e(t)$ is scheduled at time $t$ with probability $z_{v,<t}$ in the rounding algorithm, the probability that $e$ is not covered at time $t$ is precisely the probability that at most $k_e(t)-1$ vertices from $e\setminus A_e(t)$ are scheduled before time $t$. If we let $P^k(x)$ to denote an upper bound on $\Pr[S\leq k-1]$ if $kx=\E[S]$ where $S$ is the sum of Bernoulli random variables, then
\begin{equation}
\label{eq:better_bound}    
c_z(e)\approx t_e+\sum_{t>t_e}P^{k_e(t)}(z_{e,<t}+\delta(t)/k_e(t))),
\end{equation}
where $t_e$ is a special time extracted from the LP solution. Compared to previous results, where $c_z(e) \leq t_e+\sum_{t>t_e}P_c(z_{e,<t})$ (here, the quantity $P_c(z_{e,<t})$ is an upper bound on the probability if $e$ is not covered at time $t$, independent of $k_e(t)$, proved in \cite{BansalBFT21}), our expression more precisely reflects the probability of covering the hyperedge assuming that $k_e(t)$ is known. This allows for a more fine-grained analysis in terms of the change of $k_e(t)$ over time. Accordingly, we present several extended convex optimization results to argue that the worst-case behavior of GMSSC is similar to that of MSSC. Crucially, the existence of $\delta(t)$ and the properties of $P^{k_e(t)}$ allow us to prove that the worst-case scenario always involves $k_e(t)$ reducing to $1$ prior to $t_e$. On a more technical side, we give a refined upper bound on $P^k(x)$ in Theorem~\ref{thm:newtailbound}, and use these bounds to relate various terms in the expression (\ref{eq:better_bound}) for $c_z(e)$ as $t$ increases. This more refined analysis yielded a $4.509$-approximation for GMSSC.

Beyond GMSSC, we also study the min latency set cover problem. While $2$-approximation algorithms are already known for this problem, we explore whether this LP framework can yield a matching approximation. In the min latency set cover problem, the LP solver needs to schedule all vertices of a hyperedge before it is considered covered. However, a simple randomized algorithm again proves insufficient. If the LP assigns an extent of $1-1/n$ to all vertices of a hyperedge, making the edge almost scheduled fractionally, the simple randomized algorithm covers this hyperedge at time $t$ with probability $(1-1/n)^n\approx 1/e$, which is much less than $1$. Consequently, we adopt a kernel:
\[K(t,t')=\frac{\alpha t'}{t(t+1)},\]
which retains the property of increasing the probability of scheduling a vertex over time, but at a faster rate, as a divergent kernel is not required here. This approach, using a similar analysis, yields a 2-approximation. Although this result for the min latency set cover is not novel, it highlights the robustness and power of this LP-rounding framework and offers insights that could potentially lead to further improvements in the approximation factor for GMSSC.

\subsection{Organization}
Section~\ref{section:GMSSC} presents our primary analysis regarding the Generalized Minimum Sum Set Cover (GMSSC) problem. Specifically, we begin by briefly introducing the linear program, its relaxation, and the rounding algorithm from \cite{BansalBFT21} in Sections~\ref{section:LP} and~\ref{section:rounding}, respectively. Our main contribution, detailed in Section~\ref{section:analysis}, is the achievement of a $4.509$-approximation factor within this LP-rounding based framework. 

Section~\ref{section:MLC} focuses on the min latency set cover problem, where we prove that this algorithm achieves a $2$-approximation using a distinct kernel compared to the MSSC/GMSSC approach. Finally, Section~\ref{section:tailbound} establishes a new tail bound on the sum of Bernoulli random variables, which is crucial for the analysis presented in Section~\ref{section:analysis}.

\section{Generalized Min Sum Set Cover}
\label{section:GMSSC}
\subsection{Linear Program}
\label{section:LP}
Assume that time is discrete and that time $t\in\{1,2,\dots\}$, refers to the time interval $(t-1,t]$. $e$ refers to an edge and a hyperedge interchangeably. For each vertex $v$ and time $t$, there is a variable $x_{v,t}$, a 0-1 indicator of whether $v$ is assigned at time $t$. For each edge $e$ and time $t$, there is a variable $u_{e,t}$ intended to be $1$ if $e$ is uncovered at the beginning of $t$. The following LP program and relaxation for MSSC was introduced in~\cite{BansalGK10}:
\begin{alignat}{2}
    \textrm{Minimize}\quad&\sum_{e,t}u_{e,t}\quad\text{s.t.} \notag\\
    &\sum_v x_{v,t}\leq 1,\quad & \forall t, \label{eqn:cons2}\\
    &u_{e,t}+\sum_{v\in e}\sum_{t'<t}x_{v,t'}\geq 1,\quad & \forall e,t, \label{eqn:cons3}\\
    \quad&u_{e,t},x_{v,t}\geq 0,\quad & \forall e,v,t. \label{eqn:cons1}
\end{alignat}

The constraint (\ref{eqn:cons2}) ensures that at most one vertex is scheduled simultaneously. The constraint (\ref{eqn:cons3}) guarantees that $u_{e,t}$ can be set to 0 only if some $v\in e$ is scheduled strictly before $t$.
%\subsubsection{LP Relaxation for GMSSC}

For GMSSC, the demands $k_e$ are arbitrary, so the natural extension of (\ref{eqn:cons3}),
\[k_eu_{e,t}+\sum_{v\in e}\sum_{t'<t}x_{v,t'}\geq k_e,\quad \forall e,t,\]
becomes extremely weak. For example, if a hyperedge has $100$ vertices to cover, and the LP solver assigns $\sum_{t'<t}x_{v,t'}=1$ to $99$ vertices but sets all other vertices to $0$. In such a case, the natural extension of (\ref{eqn:cons3}) allows the LP solver to set $u_{e,t}=0.01$, and $e$ is almost covered in the fractional solution. However, since there are only $99$ vertices with $\sum_{t'<t}x_{v,t'}\neq0$, we are unable to obtain a valid integral solution in any case. Consequently, the Knapsack Cover (KC) Inequalities were introduced to replace (\ref{eqn:cons3}) and strengthen the linear program in~\cite{BansalGK10}.
\begin{align}
    (k_e-|S|)u_{e,t}+\sum_{v\in e\setminus S}\sum_{t'<t}x_{v,t'}\geq(k_e-|S|),\quad\forall e,t,S\subseteq e,|S|<k_e.
\label{eqn:cons4}
\end{align}
The KC inequalities guarantee that the LP solver can only consider this hyperedge $e$ covered when $k_e$ many vertices are almost scheduled.

\subsection{Rounding algorithm}
\label{section:rounding}
Once the LP solver provides the optimal fractional solution $x$, we shall apply a rounding algorithm to $x$ to generate the integral solution. However, directly rounding $\sum_{t'<t}x_{v,t'}$ may not guarantee a feasible integral solution. For example, when $k_e=1$, the LP solver can set $\sum_{t'\leq t}x_{v,t'}=1/n\rightarrow 0$ for all $v\in e$. The challenge is that we may not schedule any vertex for a long period in the rounding procedure, since the probability is negligible. Therefore, we apply an additional transformation to $x$ to ensure that the integral solution gives a feasible schedule after the LP solver covers the edge. For all $v$ and $t$, we apply the following transformation,
    \begin{align}
    z_{v,t} = \sum_{t'} K(t,t')x_{v,t'}.
    \end{align}
    
The new solution $z=Kx$ will no longer satisfy the constraints in the original linear program, but $z$ can be rounded to a feasible schedule if $K$ is chosen properly. The rounding algorithm is as follows:

\begin{algorithm}[ht]
    \caption{ Rounding Algorithm from \cite{BansalBFT21}}
    \begin{algorithmic}
        \Procedure{KERNEL $\alpha$-point Rounding}{$H=(V,E),k:E\rightarrow\mathbb{N}$}
        \State $x$ $\gets$ An optimum fractional schedule for $H,k$
        \State $z_v \gets Kx_v\quad v\in V$
        \For{$v\in V$}
            \State $\alpha_v\sim$ $\textit{uniform}[0,1]$
            \State $\tau_{v,t}\gets\mathds{1}\text{[$t$=the earliest time for which $\sum_{t'<t}z_{v,t'}\geq\alpha_v], \forall t$}$.
        \EndFor
        \Return ordering $\sigma$: Scheduling vertices according to $\tau$, breaking ties at random.
        \EndProcedure
    \end{algorithmic}
    \label{alg:round}
\end{algorithm}

After the rounding procedure, the approximation algorithm gives the desired integral solution $\sigma$.

\subsection{Analysis}
\label{section:analysis}
For a vector $x_v$, define 
\[x_{v,<t}:=\sum_{t'<t}x_{v,t'}.\]
For the GMSSC problem, we select the kernel,
\begin{align}
    K(t,t')=\frac{\beta}{t}\cdot \mathds{1}[t'\leq t].
\end{align}
Note that $\beta=2.043$ in our final setting. Correspondingly,
\begin{align}
\label{eq:zvt_to_xvt}
    z_{v,t}=\frac{\beta}{t}\sum_{t'\leq t}x_{v,t'}\quad\text{and}\quad z_{v,<t}=\beta\sum_{t''< t}\frac{x_{v,\leq t''}}{t''}=\beta \sum_{t'\leq t}x_{v,t'}\sum_{t'\leq t''}^{t}\frac{1}{t''}\geq \beta\sum_{t'\leq t}x_{v,t'}\ln{\frac{t}{t'}}.
\end{align}
This kernel ensures that even if $x_{v,<t}\rightarrow0$, $z_{v,<T}$ eventually reach $1$ for some large $T$.

Fix an edge $e$ and denote the cover time of $e$ in the LP as $c_x(e)$
\[c_x(e):=\sum _{t}u_{e,t}.\]
Let $x_{e,t}=u_{e,t}-u_{e,t+1}$, then $c_x(e)$ can be written as
\[c_x(e)=\sum _{t}u_{e,t}=\sum_{t}t x_{e,t}.\]
Additionally, define $\Cov_{\sigma}(e)$ as the cover time of $e$ in $\sigma$. Similarly, $\Cov_{\tau}(e)$ is the cover time of $e$ in $\tau$. 

According to the rounding procedure, each vertex $v\in e$ will be scheduled with probability $\min{(z_{v, <t},1)}$ at time $t$. Let $y_{v,t}=\min{(1,z_{v,<t})}$, then for each vertex $v$, the event whether $v$ will be scheduled can be considered as a Bernoulli random variable with success probability $Y_{v,t}\sim B(y_{v,t})$. Furthermore, the event that edge $e$ is scheduled at time $t$ is determined by the sum of $|e|$ independent Bernoulli random variables. Let $p_t(e)$ be the probability that $e$ is not scheduled at time $t$ in $\tau$. By the definition of $\Cov_{\tau}(e)$ and the above analysis, 
\begin{align}
\label{eq:Bernoullisum}
    \E[\Cov_{\tau}(e)]=\sum_{t}p_t(e)=\sum_{t} \Pr\left[\sum_{v\in e} Y_{v,t}\leq k_e-1\right].
\end{align}
Let $S=\sum_{v\in e}Y_{v,t}$, then $\E[S]=\sum_{v \in e} y_{v,t}$. However, it is a challenge to analyze all vertices in $e$ simultaneously.  The Hoeffding inequality, in \cite{Hoeffding56}, gives an upper bound on (\ref{eq:Bernoullisum}).
\begin{lemma}
\label{lemma:Hoeffdingbound}
    If $S$ is sum of $n$ i.i.d Bernoulli random variables and $\E[S]=np$, then 
    \begin{alignat*}{2}
        \Pr[S\leq c]\leq &\sum_{i=0}^{c}\binom{n}{i}p^i\left(1-p\right)^{n-i},\quad&&\text{if $0\leq c\leq np-1$},\\
        \Pr[S\leq c]\leq&\max_{0\leq s\leq c}{\sum_{i=0}^{c-s}\binom{n-s}{i}\left(\frac{np-s}{n-s}\right)^i\left(1-\frac{np-s}{n-s}\right)^{n-s-i}},\quad&&\text{if $np-1< c <np$},\\
        \Pr[S\leq c]\leq &1,\quad &&\text{if $np\leq c\leq n$}.
    \end{alignat*}
\end{lemma}
This lemma allows us to describe the sum of Bernoulli random variables using $\E[S]$. Remark that the bounds in Lemma~\ref{lemma:Hoeffdingbound} can be further bounded by some exponential bound derived from Poisson distributions.

Up to this point, the above analysis is virtually identical to the framework in \cite{BansalBFT21}. However, there are some gaps in the framework that need to be addressed. Let $z_{e}=Kx_e$, and consider the following KC inequalities 
    \[(k_e-|S|)u_{e,t}+\sum_{v\in e\setminus S}x_{v,<t}\geq k_e-|S|.\]
Apply the kernel $K$ to both sides, we have
    \[\sum_{v\in e\setminus S} z_{v,<t}\geq (k_e-|S|)z_{e,<t}.\]
Since $z_{v,<t}$ corresponds to $y_{v,t}$ and $z_{e,<t}$ corresponds to $x_{e,t}$, the KC inequalities potentially allow us to bound $\E[S]$ with $z_{e,t}$ or $x_{e,t}$. However, if $z_{v,<t} > 1$ for some vertices, this approach fails as $\sum_{v\in e}z_{v,<t}\neq \sum_{v\in e}y_{v,t}$. Thus, fix a time $t$, define $A_{e}(t)=\{v~|~z_{v,<t}\geq1\}$, $B_e(t)=e\setminus A_{e}(t)$, and let $k_e(t)=|B_e(t)|$. Note that the vertices in $A_e(t)$ are definitely scheduled before time $t$ in the schedule $\tau$. We have
\[\sum_{v\in e\setminus A_{e}(t)} z_{v,<t}\geq (k_e-|A_{e}(t)|)z_{e,<t}.\]
Or equivalently, 
\begin{align}
\label{eq:zvt_to_zet}
    \sum_{v\in B_{e}(t)} z_{v,<t}\geq k_e(t)\cdot z_{e,<t}.
\end{align}
The challenge arises from the fact that $k_e(t)$ will decrease over time, so even the bound in Lemma~\ref{lemma:Hoeffdingbound} changes. In \cite{BansalBFT21}, the authors gave a common upper bound on $\Pr[S\leq k-1]$ for all $k$ as a function of $\frac{\E[S]}{k}$. Since it is independent of $k$, it is always "safe" as $k_e(t)$ changes. However, such a safe bound is not tight for any $k_e(t)$ and introduces some loss. To improve, we need to process the KC inequalities with caution and extract more information from them.

First, for any fixed $k_e(t)$, we show the exponential upper bound for $p_t(e)$. We prove the following lemma that further upper bounds Lemma~\ref{lemma:Hoeffdingbound}.

\begin{restatable}{theorem}{tail}
\label{thm:newtailbound}
Let $S=\sum Y_{v}$, for any integer $k\geq1$, then function 
\begin{align*}
    P^k(x)=\begin{cases}
    e^{-\sqrt{k}(x-\frac{k-1}{k})}, & \text{if $x> 1$},\\
    1, &\text{if $0\leq x\leq 1$}.
\end{cases} 
\end{align*}
is a non-increasing and convex upper bound on the probability $\Pr[S\leq k-1]$ if $kx=\E[S]$ where $S$ is the sum of Bernoulli random variables.
\end{restatable}

\begin{proof}
    Due to space constraints, the proof is omitted and can be found in the full version of this paper. The outline of this proof is that Lemma~\ref{lemma:Hoeffdingbound} gives the Binomial distributions as an upper bound of $p_t(e)$, and the Binomial distribution can be further approximated by the Poisson distribution. 
    \[\sum_{i=0}^{k-1}\binom{n}{i}p^i(1-p)^{n-i}\leq e^{-\lambda}\sum_{i=0}^{k-1}\frac{\lambda^i}{i!}.\]
    for some $\lambda$. Hence, we compare this sum of the Poisson distributions with the expression in the theorem. Using some approximation, we conclude this upper bound involving $k$.
    % We show a detailed proof of this bound in Section~\ref{section:tailbound}.
 \end{proof}   
    We note on some important properties of $P^k(.)$ here.
    \begin{enumerate}
        \item $P^k(x)$ is non-increasing and convex on $[1,\infty)$ for all $k\geq 1$.
        \item $P^1(1)< P^2(1)< P^3(1)<\dots$ and $P^1(2)> P^2(2)>P^3(2)>\dots$.
        \item $P^k(x)\leq 1$ for all $x$ and $k$.
    \end{enumerate}

Now Theorem~\ref{thm:newtailbound} can serve as an upper bound on $p_t(e)$, i.e., 
\begin{align*}
    p_t(e)\leq\Pr\left[\sum_{v\in B_t(e)}Y_{v,t}\leq k_e(t)-1\right]\leq P^{k_e(t)}\left(\frac{\sum_{v \in B_e(t)}y_{v,t}}{k_e(t)}\right).
\end{align*}

Define $t_e$ as the earliest time such that $z_{e,\leq t_e}\geq1$, and let 
\[c_z(e)=t_e+\sum_{t>t_e}p_t(e).\]
Intuitively, we only study the probability $p_t(e)$ when $t>t_e$, and assume that for any time $t\leq t_e$ the edge $e$ will not be covered (i.e., $p_t(e)=1$ for all $t\leq t_e$). Hence, by definition, $c_z(e)$ is an upper bound on $\E[\Cov_{\tau}(e)]$, that is, $\E[\Cov_{\tau}(e)]\leq c_z(e)$.

Using~(\ref{eq:zvt_to_zet}), we can apparently state that
\[p_t(e)\leq P^{k_e(t)}(z_{e,<t}),\quad\text{and}\quad c_z(e)\leq t_e+\sum_{t>t_e}P^{k_e(t)}(z_{e,<t}).\]
However, consider an illustrative example with $k_e=3$, we draw $P^3(x)$, $P^2(x)$, and $P^1(x)$ in Figure~\ref{fig:graph_of_P}. The black curve is the upper bound on $p_t(e)$ in~\cite{BansalBFT21} as a reference.

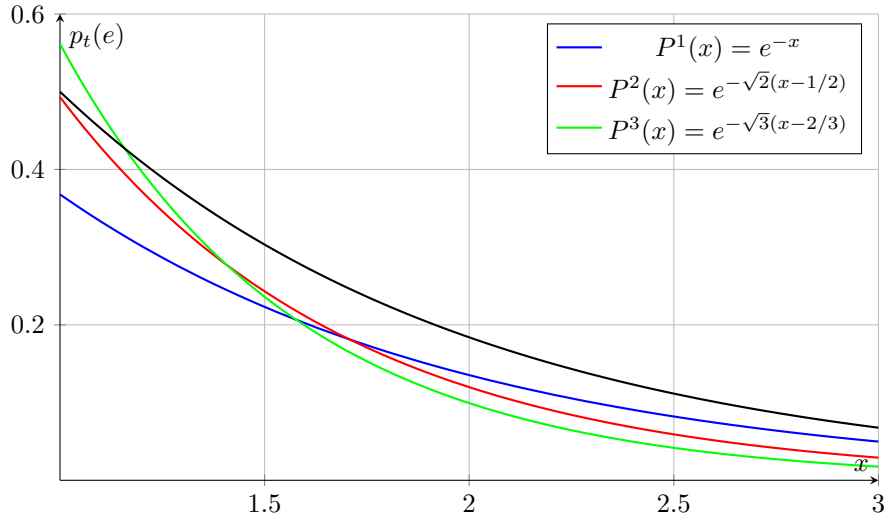
\begin{figure}[ht]
    \centering
    \begin{tikzpicture}
    \begin{axis}[
        width=0.8\textwidth,
        height=0.5\textwidth,
        xlabel=$x$,
        ylabel=$p_t(e)$,
        axis lines=middle, % or 'left', 'bottom', 'box'
        xmin=1, xmax=3,
        ymin=0, ymax=0.6,
        xtick distance=0.5,
        ytick distance=0.2,
        grid=major,
        domain=1:3, % Defines the range of x for the plot
        samples=100 % Number of points to sample the function
    ]
        \addplot[blue, thick] {e^(-x)};
        \addlegendentry{$P^1(x)=e^{-x}$}
        \addplot[red, thick] {exp((-sqrt(2)*(x-1/2)))};
        \addlegendentry{$P^2(x)=e^{-\sqrt{2}(x-1/2)}$}
        \addplot[green, thick] {exp(-sqrt(3)*(x-2/3))};
        \addlegendentry{$P^3(x)=e^{-\sqrt{3}(x-2/3)}$}
        \addplot[black, thick] {exp(1-x)/2};
    \end{axis}
\end{tikzpicture}   
\caption{Function graph of $P^i(x)$}
\label{fig:graph_of_P}
\end{figure}

The worst case for $c_z(e)$ seems to be the case that probabilities are taken subsequently from $P^3(x)$, $P^2(x)$, and $P^1(x)$, corresponding to the uppermost curve for any $x$, which implies that we switch between different $P^k(x)$ twice. The problem is whether such "jumps" are "lossless" and yield the maximum value of $c_z(e)$. If we consider a scenario in which the jump from $P^2(.)$ to $P^1(.)$ occurs at $z_{e,<t}=2.5$, we observe a probability $p_t(e)$ rise as $P^1(2.5)>P^2(2.5)$. By the definition of $p_t(e)$, its value never increases when $z_{v,<t}$ increases. Hence, there must be some additional terms to strengthen (\ref{eq:zvt_to_zet}) that we can extract from other KC inequalities. Next, we use multiple KC inequalities to approximate $\sum_{v\in B_e(t)} z_{v,<t}$.

\begin{lemma}
\label{lemma:E_cov_tau_and_c_z(e)}
    $c_z(e)\leq e^\frac{1}{\beta}(1+\frac{1}{\beta-1}e^{-1})c_x(e).$
\end{lemma}
\begin{proof}
Define a series of time slots $\{t_1,t_2,\dots t_{k_e-1}\}$ such that $t_i$ is the earliest time that there are $k_e-i$ vertices, denoted $Q_i=\{q_{k_e},\dots, q_{i+1}\}$, satisfying $z_{v,<t_i}\geq 1$. In other words, $Q_i = A_{e}(t_i)$, $|Q_i| = k_e-i$ and $Q_i\setminus Q_{i+1} = q_{i+1}$. From time $t_{i+1}$ to $t_{i}$, the cover requirement for edge $e$ changes from $i+1$ to $i$. Therefore, for $t\in (t_{i+1}, t_i)$,
\[p_t(e)\leq P^{i+1}\left(\frac{\sum_{v\in e\setminus Q_{i+1}}y_{v,t}}{i+1}\right)\]
Since $t_e\leq t_{k_e-1}\leq \dots\leq t_2\leq t_1$, in general $c_z(e)$ is at most

\begin{align}
\label{eq:cze}
c_z(e)&\leq t_e+\sum_{t=t_e}^{t_{k_e-1}-1}P^{k_e}\left(\frac{\sum_{v\in e}y_{v,t}}{k_e}\right)\\&+\sum_{t=t_{k_e-1}}^{t_{k_e-2}-1}P^{k_e-1}\left(\frac{\sum_{v\in e\setminus Q_{k_e-1}}y_{v,t}}{k_e-1}\right)+\dots+\sum_{t>t_1}P^1\left(\sum_{v\in e\setminus Q_1}y_{v,t}\right).
\end{align}

Now, we carefully use multiple KC inequalities to evaluate $\sum y_{v,t}$ with $z_{e,<t}$. For any $t_i>t_e$, assume $z_{e,<t_i}=1+\delta_i>1$, and consider the following KC inequality at time $t_i$ with the set $S$ being $Q_{i+1}$,
\begin{align}
\label{eq:delta}
 &\sum_{v\in e\setminus Q_{i+1}}z_{v,<t_i}\geq (i+1)\cdot z_{e,<t_i},\\    \implies \quad &z_{\{q_{i+1}\},<t_i}+\sum_{v\in e\setminus Q_{i}}z_{v,<t_i}\geq (i+1)\cdot z_{e,<t_i},
\end{align}
Using the KC inequality for $x$, for any time $t$, $\sum_{v\in e\setminus Q_i}x_{v,<t}\geq i\cdot x_{e,<t}$. Therefore, for any time $t>t_i$, 
\[\sum_{v\in e\setminus Q_{i}}z_{v,t}\geq i \cdot z_{e,t},\]
which leads to 
\begin{align*}
    \sum_{v\in e\setminus Q_{i}}z_{v,<t}&=\sum_{v\in e\setminus Q_i}z_{v,<t_i}+\sum_{v\in e\setminus Q_i}\sum_{t'=t_i}^{t-1}z_{v,t'}\\
    &\geq i\cdot z_{e,<t_i}+(z_{e,<t_i}-z_{q_{i+1},<t_i})+i\sum_{t'=t_i}^{t-1}z_{e,t'}\\
    &\geq i\cdot z_{e,<t}+\delta_i,
\end{align*}
where we used the fact that $z_{q_{i+1},<t_i}<1$ and $z_{e,<t_i}=1+\delta_i$.

Hence, as $t_i\geq t_e$, there exists a non-decreasing gap between $\sum_{v\in e\setminus Q_{i}}{z_{v,<t}}$ and $iz_{e,<t}$. Moreover, these gaps can accumulate over time. For example, at any time $t>t_{k_e-1}$, we have
\[\sum_{v\in e\setminus \{q_{k_e-1}\}}z_{v,<t}\geq (k_e-1)z_{e,<t}+\delta_{k_e-1}.\]
At time $t\geq t_{k_e-2}$, using the above inequality, we obtain
\[z_{q_{k_e-2},<t}+\sum_{v\in e\setminus Q_{k_e-2}}z_{v,<t}\geq (k_e-1)z_{e,<t}+\delta_{k_e-1}.\]
Applying the argument above, we get
\[\sum_{v\in e\setminus Q_{k_e-2}}z_{v,<t}\geq (k_e-2)z_{e,<t}+\delta_{k_e-2}+\delta_{k_e-1}.\]
Therefore, for any time $t_{i+1}\leq t<t_i$, we use the following inequality to approximate the probability that $e$ is not covered before time $t$.
\[p_t(e)\leq P^{i+1}\left(z_{e,<t}+\frac{\sum_{j={i+1}}^{k_e-1}\delta_{j}}{i+1}\right)\]

Hence, the upper bound~(\ref{eq:cze}) can be written as
\begin{align}
\label{eq:cze_new}
\E[\Cov_{\tau}(e)]&\leq t_e+\sum_{t=t_e}^{t_{k_e-1}-1}P^{k_e}\left(z_{e,<t}\right)+\dots+\sum_{t>t_2}^{t_1}P^2\left(z_{e,<t}+\frac{\sum_{j=2}^{k_e-1}\delta_j}{2}\right)\\
&\qquad +\sum_{t>t_1}P^1\left(z_{e,<t}+\sum_{j=1}^{k_e-1}\delta_j\right)\notag.
\end{align}
Doing similar calculations as in~(\ref{eq:zvt_to_xvt}), but for $e$ instead of $v$, we get a similar formula for $z_{e,<t}$, 
\[z_{e,<t}\geq \beta\sum_{t'\leq t}x_{e,t'}\ln{\frac{t}{t'}}.\]
Let $t^*$ be the (fractional) time such that
\[\beta\sum_{t'\leq \floor{t^*}}x_{e,t}\ln{t^*/t'}=1.\]
Naturally, $z_{e,<t*}\geq 1$. Since $t_e$ is the earliest time that $z_{e,\leq t_e}\geq1$, we have $\floor{t^*}\geq t_e$. Let $\gamma_t=\beta(\sum_{t'\leq t}{a_{t'}\ln{(t/t')}})$ where $a_t$ corresponds to $x_{e,t}$ in our analysis for GMSSC. Using the non-increasing property of $P^k$, substituting $z_{e,<t}$ with $\gamma_t$ only increases $c_z(e)$. Therefore, we formulate $c_z(e)/c_x(e)$ as the following optimization problem. Recall that $c_x(e)=\sum_{t}tx_{e,t}$.
\begin{align}
    &(\mathcal{F})\quad\textrm{Maximize}\ \frac{\floor{t^*}+\sum_{t=\floor{t^*}+1}^{t_{k_e-1}-1}P^{k_e}(\gamma_t)+\cdots+\sum_{t_1}P^1(\gamma_t+\sum_j {\delta_j})}{\sum_t ta_t}\\
    &\text{s.t.}\quad\pnorm{a}{1}=1,\quad \gamma_{t^*}=1,\quad (z_{q_{i+1},< t_{i}}<1, \text{ and } z_{q_{i+1},\leq t_{i}}>1) \dots \notag
\end{align}
For simplicity, we denote $\sum_{t>\floor{t^*}}P^{k_e(t)}(\gamma_t+\sum_{j=k_e(t)}^{k_e-1} {\delta_j}/k_e(t))$ as $\sum_{t>t^*}P(\gamma_t)$, where $P(x)\leq 1$, $P'(x)\leq 0$ and $P''(x)\geq 0$. 

The following fact uses the convexity of a function to conclude its maximum value on a compact region.
\begin{fact}[\cite{BansalBFT21} Fact 8]
\label{fact:compact}
Let $f: R^n\rightarrow R$ be a convex and non-negative function on some compact region $H$, and $g:R^n\rightarrow R$ is a linear and positive function on $H$, then the maximum of $\max_{x\in H}f(x)/g(x)$ is attained at an extreme point of $H$.
\end{fact}

\begin{proof}
    Consider any $x,y\in H$ and $\lambda\in[0,1]$,
    \[\max\left(\frac{f(x)}{g(x)},\frac{f(y)}{g(y)}\right)\geq\frac{\lambda f(x)+(1-\lambda)f(y)}{\lambda g(x)+(1-\lambda)g(y)}\geq \frac{f(\lambda x+(1-\lambda)y)}{g(\lambda x+(1-\lambda) y)}.\]
\end{proof}

The following lemma describes the maximum scenario for $\mathcal{F}$.
\begin{lemma}
\label{lemma:max_point}
    If $P(x)$ is function such that $P(x)\leq 1$, $P''(x)\geq 0$, and $P'(x)\leq 0$ for all $x\geq1$, then for the maximum problem
    \[\mathcal{F}=\frac{\floor{t^*}+\sum_{t>t^*}P(\gamma_t)}{\sum_{t}t a_{t}},\] with at least $m\geq 2$ non-trivial constraints, it always attains its maximum value when some $a_{t}=1$.  
\end{lemma}

\begin{proof}

In \cite{BansalBFT21}, the authors proved this lemma when $m=2$. Let us consider the case where there are more non-zero variables. We focus on how to reduce the number of nonzero variables. Since each $P^i$ is convex, $\sum_{t}P^{k_e(t)}$ is also a convex function. Using Fact~\ref{fact:compact}, the max value of $\mathcal{F}$ attains at some extreme point, where there are at most $m$ non-zero variables.

Assume that there non-zero variables are $c_1\leq c_2\leq\dots\leq c_m$ and corresponding values are $a_1,a_2,\dots, a_m$. The following two cases discuss the non-zero variables before $t^*$ or after $t^*$, respectively.
\begin{enumerate}
    \item $c_1\leq \dots c_{l-1} \leq t^*$. 
    
    Suppose that $\sum_{i=1}^{l-1}a_i=D$ for some constant $D\leq 1$ and let $c_1=t^*e^{-b_1}$, $c_2=t^*e^{-b_2}$, up to $c_{l-1}=t^*e^{-b_{l-1}}$. The constraint $\gamma_{t^*}=1$ becomes 
    \[\beta\sum_{i=1}^{l-1} a_ib_i=1.\]
    
    For any time $t\geq t^*$, 
    \[\gamma_t=D\beta\ln{t/t^*}+\beta\sum_{i=1}^{l-1} a_ib_i+\sum_{i=l}a_i\ln{t/c_i}=1+D\beta\ln{t/t^*}+\sum_{i=l}a_i\ln{t/c_i}. \]
    Which indicates that the numerator is independent of any individual $a_i$ or $b_i$ for $i < l$. Then $\mathcal{F}$ becomes a minimization problem. In particular,
    \[\text{Minimize}\quad t^*\beta\left(\sum_{i=1}^{l-1} a_ie^{-b_i}\right)+\sum_{i=l}c_ia_i\quad\text{s.t.} \quad \beta\left(\sum_{i=1}^{l-1} a_ib_i\right)=1.\]
    Using the convexity of $e^{-x}$, the extreme point is at $b_1=b_2=\cdots=b_{l-1}$. Consequently, we reduce the number of non-zero variables before $t^*$ to $1$. 
    
    \item $t^*\leq c_l\leq \dots\leq c_m$.

    Now we assume that $b_1=b_2=\cdots=b_{l-1}=B$, the constraint $\gamma_{t^*}=1$ becomes
    \[DB\beta=1,\] 
    and is independent of $c_l,\dots, c_m$. Let $c_l=t^*b_l$, $c_{l+1}=t^*b_{l+1}$, up to $c_m=t^*b_m$. For time $t>t^*$, the denominator is
    \begin{align}
    \label{eq:denominator}
        t^*e^{-B}D\beta + t^* \sum_{i=l} b_i a_i.
    \end{align}
    While the numerator is
    \begin{align*}
    &t^*+\sum_{t=t^*+1}^{c_1}P(1+D\beta\ln{t/t^*})+\sum_{c_1+1}^{c_2}P(1+D\beta\ln{t/t^*}+\beta a_l \ln{t/t^*a})\\&+\cdots+\sum_{c_m}^{\infty}P(1+D\beta\ln{t/t^*}+\beta\sum_{i=l}^m a_i \ln t/t^*b_i).
    \end{align*}
    
     Let $x=t/t^*$, we can upper bound the numerator using the fact that all $P^i$ are non-increasing and $P^i(\gamma_t)\leq 1$. For convenience, we denote $A_m(x)=1+D\beta\ln x+\beta\sum_{i=l}^m a_i (\ln{x}-\ln{b_i})$. Hence, 
    \begin{align}
    \label{eq:numerator}
        &t^*\left(1+\int_{1}^{b_l}P(1+D\beta\ln x)dx+\cdots+\int_{b_{m-1}}^{b_m}P(A_{m-1}(x))dx\right) \notag\\
        &+t^*\left(\int_{b_m}^{\infty}P(A_{m-1}(x)+\beta a_m(\ln{x}-\ln{b_m}))dx\right).
    \end{align}

    Notice that only the last two terms contain $b_m$, so we first try to eliminate $b_m$. Let (\ref{eq:numerator}) be $f(b_m)$ and let (\ref{eq:denominator}) be $g(b_m)$, we show that $f(b_m)/g(b_m)$ always reaches its maximum value when $b_m=b_{m-1}$ and thus eliminate $b_m$ and $c_m$.
    \end{enumerate}

    \begin{claim}
    \label{claim:max_point}
        The function $f(b_m)/g(b_m)$ achieves its maximum value at $b_m = b_{m-1}$.
    \end{claim}

    \begin{proof}
    We rename $b_m$ as $b$ here. To prove this claim, we show that $(f(b)/g(b))'\leq 0$ for all $b \geq b_{m-1}$. Define 
    \[h(b)=f'(b)g(b)-f(b)'g(b).\]
    and correspondingly, 
    \[h'(b)=f''(b)g(b)-f(b)g''(b).\]
    It suffices to show that $h(b)\leq 0$ for all $b\geq b_{m-1}$. Note that $g'(b)=a_m$, and $g''(b)=0$. Therefore, $h'=f''g$, and the sign of $h'$ depends only on $f''$. We then compute $f'(b)$ and $f''(b)$.
        \begin{align*}
        f'(b)&=P\left(A_{m-1}(b)\right)-P\left(A_{m-1}(b)+\ln{b}-\ln{b}\right)\\
        &-\frac{\beta}{b}a_m\int_b^\infty P'\left(1+D\beta\ln x+\sum_{i=l}^{m}a_i(\ln x-\ln b_i)\right)dx\\
        &\mbox{substituting x=bt}\\
        &=-\beta a_m \int^\infty_1 P'\left(1+D\beta\ln t+D\beta \ln b+\sum_{i=l}^{m-1}a_i(\ln b+\ln t-\ln b_i)+a_m\ln t\right)dt.
        \end{align*}
    Hence,
    \begin{align*}
        f''(b)=-\frac{\beta^2a_m}{b}\int^\infty_1P''(1+D\beta\ln t+D\beta \ln b+\sum_{i=l}^{m-1}a_i(\ln b+\ln t-\ln b_i)+a_m\ln t)dt.
    \end{align*}
    Because all $P^i$ is convex on $[1,\infty)$, for all $b\geq b_{m-1}$
    \[P''(1+D\beta\ln t+D\beta \ln b+\sum_{i=l}^{m-1}a_i(\ln b+\ln t-\ln b_i)+a_m\ln t)\geq 0.\] Thus, $f''(b)\leq 0$ and $h'(b)\leq 0$ for all $b\geq b_{m-1}$.

    Meanwhile, we need to show that $h(b_{m-1})\leq 0$. However, proving $h(b_{m-1})\leq 0$ is equivalent to solving a subproblem of this claim, particularly, exactly one less non-zero variable after $t^*$. Therefore, we can recursively apply the same argument and reduce the number of non-zero variables. Finally, we reach the case where there is one non-zero variable after $t^*$, and we prove $h(1)\leq 0$. We list the proof of $h(1)\leq 0$ here.
    \begin{align*}
        &f(1)=1+\int_1^{\infty}P(1+\beta\ln x)dx\quad\text{and}\quad g(1)=De^{-1/D\beta},\\
        &f'(1)=-\int_{1}^\infty\beta(1-D)P'(1+\beta\ln x)dx\ \text{and}\quad g'(1)=1-D.
    \end{align*}
    Note that $f'(1)\geq0$ as $P'(x)\leq 0$ and $g(1)\leq 1$, we have
    \begin{align*}
        h(1)&=f'(1)g(1)-f(1)g'(1)\leq f'(1)-f(1)(1-D)\\
        &\leq-(1-D)\left(1+\int_{1}^\infty\beta P'(1+\beta\ln x)+P(1+\beta\ln x)dx \right)\\
        &=-(1-D)\left(1+\int_{1}^\infty \left(e^{\frac{y-1}{\beta}}P(y)\right)'dy\right)\tag{let $1+\beta\ln x=y$}\\ 
        &\leq (1-D)(1-P(1))\leq 0.
    \end{align*}
    Consequently, $h(b)$ is a non-increasing function on $[b_{m-1},\infty)$ and $h(b)\leq h(b_{m-1})$ for all $b\geq b_{m-1}$.  
    \end{proof}
    
    With Claim~\ref{claim:max_point}, we conclude that any $m\geq 3$ can be reduced to the $m=1$ case, where the remaining constraint is $\gamma_{t^*}=1$. Solving $\gamma_{t^*}=1$ when some $a_{u}=1$, we obtain $u=t^*\cdot e^{-1/\beta}$. 

\end{proof}

Using Lemma~\ref{lemma:max_point}, we conclude that the maximum value of $\mathcal{F}$ occurs at time $u=t^*\cdot e^{-1/\beta}$ such that $x_{e,u}=1$, and we simplify $\mathcal{F}$ as follow.

\begin{align}\label{eq:F_integral}
\mathcal{F}&\leq \frac{1}{u}\left(t^*+\int_{t^*}^{t_{k_e-1}}P^{k_e}\left(\beta\ln{\frac{t}{u}}\right)dt+\cdots+\int_{t_1}^{\infty}P^1\left(\beta\ln{\frac{t}{u}}+\sum_{j}\delta_j\right)dt\right)\notag\\
&\leq e^{\frac{1}{\beta}}+\frac{1}{\beta}\left(\int_{1}^{1+\delta_{k_e-1}}P^{k_e}(x)e^{\frac{x}{\beta}}dx+\cdots\right)\\
&+\frac{1}{\beta}\left(\int_{1+\delta_2}^{1+\delta_1}P^2\left(x+\frac{\sum_{j=2}^{k_e-1}\delta_j}{2}\right)e^{\frac{x}{\beta}}dx+\int^\infty_{1+\delta_1} P^1\left(x+\sum_{j=1}^{k_e-1}\delta_j\right)e^\frac{x}{\beta}dx\right).\tag{let $x=\beta\ln(t/u)$ and change of variables}
\end{align}

Now, for the integral part, we argue that it is at most $\int_1^\infty e^{-x}e^{x/\beta}dx$ by demonstrating that any jump between different $P^i$ only reduces the value of $\mathcal{F}$. 
\begin{figure}[ht]
    \centering
    \begin{tikzpicture}
    \begin{axis}[
        xlabel=$x$,
        ylabel=$p_t(e)$,
        axis lines = box, % or 'left', 'bottom', 'box'
        xmin=1, xmax=3,
        ymin=0, ymax=1,
        xtick = {1,1.5,2,2.5,3},
        ytick = {0,0.2,0.4,0.6,0.8,1},
        grid=major,
        domain=1:3, % Defines the range of x for the plot
        samples=100 % Number of points to sample the function
    ]
        \addplot[blue, thick, name path=f] {e^(-x/2)};
        \addlegendentry{$P^1(x)e^\frac{x}{\beta}$}
        \addplot[red, thick, domain = 1:1.2, name path =g] {exp((-sqrt(2)*(x-1/2))+x/2)};
        \addlegendentry{$P^2(x)e^\frac{x}{\beta}$}
        \addplot[green, thick, domain=1.2:3, name path=h] {exp(-x/2-0.2)};
        \addlegendentry{$P^1(x+\delta)e^\frac{x}{\beta}$}
        \addplot[dotted, thick, black, domain=1:1.2] {exp(-x/2-0.2)};
        \addplot[dotted, thick, black, domain=1.2:3] {exp((-sqrt(2)*(x-1/2))+x/2)};
        \addplot[dashed, thick, gray, forget plot] coordinates {(1.2,0) (1.2,0.6)};
        \node[above] at (axis cs:1.35, 0.2) {$1+\delta_1$};
        \path[name path=x_axis] (axis cs:1,0) -- (axis cs:3,0); % Extend across your x-range
        
        \addplot [
        fill=gray!50,
        samples=20,
        forget plot
    ]
    fill between [
        of=f and x_axis,
        soft clip={domain=1:3},
    ];
        
        \addplot [
        color=blue!30,
        fill opacity=0.7,
        domain=1:1.2,
        samples=100,
    ]
    fill between [
        of=f and g,
        soft clip={domain=1:1.2}, 
    ];
        \addplot [
        color=red!30,
        fill opacity=0.7,
        domain=1:1.2,
        samples=100,
    ]
    fill between [
        of=f and h,
        soft clip={domain=1.2:3}, 
    ];
    \end{axis}
\end{tikzpicture}   
\caption{The loss of jumping between $P^i$}
\label{fig:loss}
\end{figure}
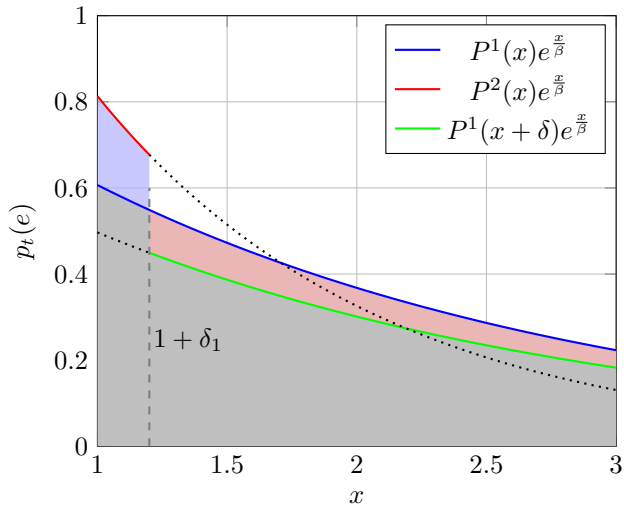

Consider $k_e=2$, and according to the analysis above, our goal is to show that
\[\int_1^{1+\delta_1}P^2(x)e^\frac{x}{\beta}dx+\int^\infty_{1+\delta_1}P^1(x+\delta_1)e^\frac{x}{\beta}dx\leq\int_{1}^\infty P^1(x)e^\frac{x}{\beta}dx.\]
We draw three curves $P^2(x)\exp(x/\beta)$, $P^1(x)\exp(x/\beta)$ and $P^1(x+\delta_1)\exp(x/\beta)$ in Fig~\ref{fig:loss}. Comparing LHS with RHS (gray area), we observe two differences: the blue area represents the gain from $\int^{1+\delta_1}_1 P^2(x)^\frac{x}{\beta}dx$, and the red area represents the loss due to the gap $\delta_1$. If we can prove that the red area is larger than the blue one, we can argue that the jump ($\delta_1>0$) introduces some loss, and the maximum value returned from $\int_1^\infty P^1(x)^\frac{x}{\beta}dx$.

The following claim formally proves that the worst case is always $P^1(.)$ due to these gaps $\delta_i$.

\begin{claim}
\label{claim:max_integral}
For any $b\geq  a\geq0$ and $\Delta\geq 0$, 
    \[\int\limits_{1+a}^{1+b}P^i\left(x+\frac{a+\Delta}{i}\right)e^\frac{x}{\beta}dx+\int\limits_{1+b}^\infty P^1\left(x+b+\frac{a+\Delta}{i-1}\right)e^{\frac{x}{\beta}}dx\leq \int\limits_{1+a}^\infty P^1\left(x+a+\frac{\Delta}{i}\right)e^\frac{x}{\beta}dx.\]
\end{claim}

\begin{proof}
Let LHS be a function of $b$,
\[L(b)=\int_{1+a}^{1+b}P^i\left(x+\frac{a+\Delta}{i}\right)e^\frac{x}{\beta}dx+\int_{1+b}^\infty P^1\left(x+b+\frac{a+\Delta}{i-1}\right)e^{\frac{x}{\beta}}dx.\]
When $b=a$, as $P^i$ is a non-increasing function,
\begin{align*}
    L(a)=\int_{1+a}^\infty P^1\left(x+\frac{a+\Delta}{i-1}+a\right)dx\leq\int_{1+a}^\infty P^1\left(x+\frac{\Delta}{i}+a\right)dx=RHS.
\end{align*}
Meanwhile when $b\rightarrow\infty$,
\begin{align*}
        \lim_{b\rightarrow\infty}L(b)&=\int_{1+a}^\infty P^i\left(x+\frac{\Delta+a}{i}\right)e^{\frac{x}{\beta}}dx\\
        &=\frac{1}{\sqrt{i}-1/\beta}\exp\left(-\frac{\Delta+(i+1)a+1}{\sqrt{i}}+\frac{1+a}{\beta}\right)\\
        &=\frac{1}{\sqrt{i}-1/\beta}e^{-\frac{1}{\sqrt{i}}} e^{\frac{1+a}{\beta}} \exp\left(-\frac{\Delta+(i+1)a}{\sqrt{i}}\right).
\end{align*}
Given the facts that
\[\frac{1}{\sqrt{i}-1/\beta}e^{-\frac{1}{\sqrt{i}}}\leq \frac{1}{1-1/\beta}e^{-1},\quad \text{and}\quad \frac{i+1}{\sqrt{i}}\geq 2.\]
We conclude that 
\[\lim_{b\rightarrow\infty}L(b)\leq \frac{\beta}{\beta-1}e^{-1-2a-\frac{\Delta}{i}}e^\frac{1+a}{\beta}=\int_{1+a}^\infty P^1\left(x+a+\frac{\Delta}{i}\right)e^\frac{x}{\beta}dx.\]
The claim holds for both $b=a$ and as $b\rightarrow\infty$. Then we check if $L$ has any extreme point on $[a,\infty)$. Compute the derivative of $L(b)$,
\begin{align*}
    L'(b)&=\exp\left(-\sqrt{i}(b+\frac{1+\Delta+a}{i})+\frac{1+b}{\beta}\right)\\
    &-\frac{2\beta-1}{\beta-1}\exp\left(-(1+2b+\frac{\Delta+a}{i-1})+\frac{1+b}{\beta}\right).
\end{align*}
Let $E=(2\beta-1)/(\beta-1)$, and note that $E>e$. Solving for the extreme point $b_0$ from $L'(b_0)=0$, we get
\begin{align}
\label{eq:extreme_int}
    (\sqrt{i}-2)b_0=1-\ln\frac{2\beta-1}{\beta-1}-\frac{1}{\sqrt{i}}+\left(\frac{1}{i-1}-\frac{1}{\sqrt{i}}\right)(\Delta+a).
\end{align}
We discuss the extreme point based on the value of $i$.
\begin{enumerate}
    \item $i>4$. In this case, the RHS of (\ref{eq:extreme_int}) is a decreasing function of $\Delta$ and $a$,  as $1/(i-1)<1/\sqrt{i}$ for all $i\geq 3$. Hence,
    \[b_0 < -\frac{1}{\sqrt{i}(\sqrt{i}-2)}<0.\]
    Thus, $L(b)$ is monotone on $[0,\infty)$. The maximum of $L(b)$ is either $L(a)$ or $L(\infty)$, and the claim holds.
    % $L(b)$ is nonincreasing for all $b>0$ when $i>4$. Therefore, $L(a)$ is the maximum value for all $b\geq a$.
    \item $i = 4$. Now $L(b)$ has no extreme points. Since $L$ is monotone, the claim holds.

    \item $i=3$. Potentially, $b_0>a\geq 0$, but consider $L'(0)$,
    \begin{align*}
        L'(0)&<e^{-\frac{1+\Delta+a}{\sqrt{i}}+\frac{1}{\beta}}-e\cdot e^{-1-\frac{\Delta+a}{i-1}+\frac{1}{\beta}}\\
        &\leq e^\frac{1}{\beta} \left(e^{-\frac{1}{\sqrt{i}}}e^{-\frac{\Delta+a}{\sqrt{i}}}-e^{-\frac{\Delta+a}{i-1}}\right)< 0.
    \end{align*}
    Hence, $L(b_0)$ gives the minimum value. The maximum value is $L(a)$ or $L(\infty)$, so the claim holds.
    \item $i=2$. We compute the extreme point,
    \[b_0=-\frac{1}{2}(\Delta+a+1)+\frac{1}{2-\sqrt{2}}\ln E.\]
    Let $b_0\geq a$ so that there is an extreme point on $[a,\infty)$, we have,
    \[\Delta+3a+1\leq\frac{2}{2-\sqrt{2}}\ln E.\]
    Consider $L'(a)$
    \[L'(a)=e^{\frac{1+a}{\beta}}(e^{-\frac{\sqrt{2}}{2}(1+3a+\Delta)}-e^{\ln E-1-\Delta-3a}).\]
    The sign of $L'(a)$ depends on
    \[-\frac{\sqrt{2}}{2}(1+3a+\Delta)-(\ln E-\Delta-1-3a)=\frac{2-\sqrt{2}}{2}(1+3a+\Delta)-\ln E\leq 0.\]
    As a result, $L'(a)\leq 0$, and $L(b_0)$ again gives the minimum value. The maximum value is $L(a)$ or $L(\infty)$, and the claim is valid.
\end{enumerate}
Towards upper bounding~(\ref{eq:F_integral}), we set $a=\delta_{i}$, $b=\delta_{i-1}$, and $\Delta$ represents the sum of all the other $\delta_j$. 
\end{proof}

Using Claim~\ref{claim:max_integral} recursively, we obtain

\begin{align*}
    \mathcal{F}&\leq e^{\frac{1}{\beta}}+\frac{1}{\beta}\left(\cdots+\int_{1+\delta_2}^{1+\delta_1}P^2\left(x+\frac{1}{2}\sum_{j=2}^{k_e-1}\delta_j\right)dx+\int^\infty_{1+\delta_1} P^1\left(x+\sum_{j=1}^{k_e-1}\delta_j\right)e^\frac{x}{\beta}dx\right)\\
    &\leq e^\frac{1}{\beta}+\frac{1}{\beta}\left(\cdots+\int^\infty_{1+\delta_2} P^1\left(x+\frac{1}{2}\sum_{j=3}^{k_e-1}\delta_j+\delta_2\right)e^\frac{x}{\beta}dx\right)\\
    &\leq e^\frac{1}{\beta}+\frac{1}{\beta}\left(\cdots+\int^\infty_{1+\delta_3} P^1\left(x+\frac{1}{3}\sum_{j=4}^{k_e-1}\delta_j+\delta_3\right)e^\frac{x}{\beta}dx\right)\\
    &\leq \cdots\\
    &\leq e^\frac{1}{\beta}+\frac{1}{\beta}\left(\int^\infty_{1} P^1(x)e^\frac{x}{\beta}dx\right)\\
    &=e^\frac{1}{\beta}\left(1+\frac{1}{\beta-1}e^{-1}\right).
\end{align*}

On the other hand, if $t_1\leq t_e$, then $c_z(e)\leq t_e+\sum_{t_e}P^{k_e-1}(z_{e,<t})+\dots$, which is equivalent to the case where there are $k_e-1$ vertices remain to be covered, and allowing us to apply the same argument. Hence,
\[c_z(e)\leq e^\frac{1}{\beta}\left(1+\frac{1}{\beta-1}e^{-1}\right)c_x(e).\]
\end{proof}

\begin{lemma}
\label{lemma:E_cov_sigma_and_E_cov_tau}
    $\E[\Cov_{\sigma}(e)]\leq \beta c_z(e)$
\end{lemma}
\begin{proof}
    Apply the kernel to (\ref{eqn:cons2}),
    \[\sum_{v}z_v=\sum_{v}Kx_v\leq K\mathbf{1}=\beta\mathbf{1}.\]
    Intuitively, the kernel allows the rounding algorithm to schedule at most $\beta$ vertices at the same time, indicating that $\E[\Cov_{\sigma}(e)]\approx\beta \E[\Cov_{\tau}(e)]$ due to the random tie-breaking rule. See~\cite{BansalBFT21} Lemma 19 for more details.
\end{proof}

\main*

\begin{proof}
With Lemma~\ref{lemma:E_cov_tau_and_c_z(e)} and Lemma~\ref{lemma:E_cov_sigma_and_E_cov_tau}, it follows that the approximation ratio is at most
\[\beta e^{\frac{1}{\beta}}\left(1+\frac{1}{\beta-1}e^{-1}\right).\]
The ratio is $4.509$ when $\beta=2.043$.
\end{proof}

\section{2-approximation for Min Latency Set Cover}
\label{section:MLC}
We also study the min latency set cover problem, which is a special case of GMSSC where $k_e=|e|$ for all hyperedges $e$. Apply $K$ to $x_e$ and $x_v$, the KC inequalities become 
\[z_{v,<t}\geq z_{e,<t},\quad \forall~v\in e.\]
Hence, by the earliest time $t_e$ such that $z_{e,\leq t_e}\geq 1$, all the vertices $v\in e$ have been scheduled in $\tau$. 

For this problem, we apply a different kernel,
\[K(t,t')=\frac{\alpha t'}{t(t+1)}\mathds{1}[t'\leq t].\]
Similar to the GMSSC case, the following two lemmas suffice to yield the desired approximation.
\begin{lemma}
    \[c_z(e)\leq\frac{\alpha}{\alpha-1} c_x(e).\]
\end{lemma}

\begin{proof}
    Under the kernel $K=\alpha t'/{t(t+1)}$, the relationship between $x$ and $z$ is
    \begin{align*}
        z_{e,<t}&=\sum_{t'\leq t}\sum_{t''\leq t'}K(t',t'')x_{v,t''}=\sum_{t''<t}\sum_{t'=t''}^{t-1}K(t,t'')=\alpha\sum_{t'=1}^t\frac{t-t'}{t}x_{e,t'},
    \end{align*}
    where we use $\sum_{q=t''}^{t-1}K(q,t'')=\alpha t''\sum_{q=t''}^{t-1}(1/q-1/(q+1))=\alpha(1-t''/t)$.

    Define $t^*$ as the time that $\alpha\sum_{t'=1}^{t^*}((t^*-t')/t^*)x_{e,t'}=1.$
    Again, $z_{e,<t^*}=1$ indicates $t^*\geq t_e$. Since $c_z(e)\leq t_e$, the corresponding maximum problem becomes
    \begin{align}
        \label{eq:MLC}
        &\text{maximize}\quad\frac{t^*}{\sum_{t}ta_{t}}, \quad \text{s.t.}\quad \pnorm{a}{1}=1,~\alpha\sum_{t\leq t^*}a_t\frac{t^*-t}{t^*}=1.\notag
    \end{align}

    This problem has two non-trivial constraints, implying it has at most 2 non-zero variables. Let $a_u=s$ and $a_v=1-s$. 
    \begin{enumerate}
        \item $u<v\leq t^*$, the problem is
        \begin{gather*}
        \text{Minimize}\quad su+(1-s)v,\\
        \text{s.t.}\quad s\frac{t^*-u}{t^*}+(1-s)\frac{t^*-v}{t^*}=1-\frac{su+(1-s)v}{t^*}=\frac{1}{\alpha}.
        \end{gather*}
        In this case, $su+(1-s)v$ is a constant $(t^*)(1-1/\alpha)$. Hence, the ratio in~(\ref{eq:MLC}) is at most
        \[\frac{\alpha t^*}{t^*(\alpha-1)}= \frac{\alpha}{\alpha-1}.\]
        \item $u<t^*<v$, then the problem is,
        \[\text{Minimize}\quad su+(1-s)v\quad, \text{s.t.}\quad s\frac{t^*-u}{t^*}=\frac{1}{\alpha},\]
        This minimization problem is simply an increasing function of $v$, which leads to $v=t^*$.
    \end{enumerate}
\end{proof}
\begin{lemma}
    \[\E[{\Cov_{\sigma}(e)}]\leq \frac{\alpha}{2}c_z(e).\]
\end{lemma}
\begin{proof}
    Apply $K$ to (\ref{eqn:cons2}), we have
    \[\sum_{v}z_{v}=\sum_v Kx_{v}\leq K\mathbf{1}.\]
    Furthermore, 
    \[(K\mathbf{1})_t=\sum_{t'}K(t,t')=\sum_{t'\leq t}\frac{\alpha t'}{t(t+1)}=\frac{\alpha}{t(t+1)}\cdot\frac{t(t+1)}{2}=\frac{\alpha}{2}.\]
    Hence, by time $t_e$, $\sigma$ schedules at most $(\alpha/2)t_e$ vertices and must include all the vertices $v\in e$, then
    \[\E[\Cov_\sigma(e)]\leq\frac{\alpha}{2}t_e=\frac{\alpha}{2}c_z(e).\]
\end{proof}

\second*
\begin{proof}
Combining the above two lemmas, we obtain the ratio for the min latency set cover problem, 
\[\frac{\alpha^2}{2(\alpha-1)}, \alpha\geq2,\]
which attains the minimum value $2$ at $\alpha=2$.
\end{proof}
\begin{remark}
    Theorem~\ref{thm:main2}, along with the $4$-approximation algorithm for MSSC, which is produced by the same algorithm with a different kernel, indicates that achieving a $4$-approximation for GMSSC within this framework using a universal kernel to all vertices is difficult.
\end{remark}

\section{A Tail Bound for the Sum of Bernoulli Random Variables}
\label{section:tailbound}
Now we prove Theorem~\ref{thm:newtailbound}.
\tail*

\begin{proof}
    Rewrite Lemma~\ref{lemma:Hoeffdingbound} as,
    \begin{align*}
        \Pr[S\leq k-1]\leq \sum_{i=0}^{k-1}\binom{n}{i}\left(\frac{kx}{n}\right)^i\left(1-\frac{kx}{n}\right)^{n-i},\quad\text{if $x\geq 1$}.
    \end{align*}
    \begin{lemma}[\cite{BansalBFT21} Lemma 29]
    For $\lambda=np=\E[S],$
    \[\sum_{i=0}^{k-1}\binom{n}{i}p^i(1-p)^{n-i}\leq e^{-\lambda}\sum_{i=0}^{k-1}\frac{\lambda^i}{i!}.\]
    \end{lemma}
    This lemma is a stronger version of the Poisson Limit Theorem, and we use it to approximate the sum of binomial distributions. When $k=1$ this lemma directly gives $e^{-x}$ and the theorem is valid. Hence, it suffices to prove the following claim:
    \begin{claim}
    For all $k\geq 2$, and $x\geq1$,
        \[\exp\left({-\sqrt{k}\left(x-\frac{k-1}{k}\right)}\right)\geq e^{-kx}\sum_{i=0}^{k-1}\frac{(kx)^i}{i!}.\]
    \end{claim}
    \begin{proof}
    Define $\eta=(k-1)/k$ and the following function, 
        \[R(x)=e^{-kx}\sum_{i=0}^{k-1}\frac{(kx)^i}{i!}\bigg/ e^{-\sqrt{k}(x-\frac{k-1}{k})}=e^{-kx+\sqrt{k}(x-\frac{k-1}{k})}\sum_{i=0}^{k-1}\frac{(kx)^i}{i!}.\]
    If we can prove that $R(x)\leq 1$ for all $x\geq \eta$, we finish the proof.  First, using the Maclaurin expansion of $e^x$, we have
        \[e^x\geq\sum_{i=0}^{k-1}\frac{x^i}{i!},\]
    Let $x=k-1$, we obtain $e^{k-1}\geq\sum_{i=0}^{k-1}(k-1)^i/i!$. Then at $x=\eta$,
        \[R(\eta)=e^{-(k-1)}\sum_{i=0}^{k-1}\frac{(k-1)^i}{i!}\leq 1.\]
    On the other hand, consider $R(\infty)$,
    \begin{align}
        \lim_{x\rightarrow\infty}R(x)=e^{-\frac{k-1}{\sqrt{k}}}\lim_{x\rightarrow\infty}\sum_{i=0}^{k-1}e^{(-k+\sqrt{k})x}\frac{(kx)^i}{i!}.
    \end{align}
    As $k\geq2$, $-k+\sqrt{k}<0$,
    \[\lim_{x\rightarrow\infty}\frac{(kx)^i}{i!\cdot e^{(k-\sqrt{k})x}}=0.\]
    Because $k$ is a constant, we conclude that $R(\infty)=0$. It remains to analyze the extreme point. Since $R(\eta)\leq 1$ and $R(\infty)<1$, we assume that there are some extreme points with the maximum value on $[\eta,\infty)$. Compute the derivative of $R(x)$,
        \begin{align*}
        R'(x)&=e^{-kx+\sqrt{k}(x-\frac{k-1}{k})}\left((-k+\sqrt{k})\left(\sum_{i=0}^{k-1}\frac{(kx)^i}{i!}\right)+k\sum_{i=0}^{k-2}\frac{(kx)^i}{i!}\right)\\
        &=\sqrt{k}\cdot e^{-kx+\sqrt{k}(x-\frac{k-1}{k})}\left(\frac{-\sqrt{k}(kx)^{k-1}}{(k-1)!}+\sum_{i=0}^{k-1}\frac{(kx)^i}{i!}\right).
        \end{align*}
    The extreme point $x_0$ satisfies
        \[\frac{\sqrt{k}(kx_0)^{k-1}}{(k-1)!}=\sum_{i=0}^{k-1}\frac{(kx_0)^i}{i!}.\]
    Hence, we can eliminate the summation in $R(x)$
    \begin{align}
    \label{eq:max_value_R}
        R(x)\leq R(x_0)=\frac{\sqrt{k}(kx_0)^{k-1}}{(k-1)!}e^{(\sqrt{k}-{k})x_0-\frac{k-1}{\sqrt{k}}}.
    \end{align}
    Consider~(\ref{eq:max_value_R}) as function of $x_0$, 
    \[S(t):=\frac{\sqrt{k}(kt)^{k-1}}{(k-1)!}e^{(\sqrt{k}-{k})t-\frac{k-1}{\sqrt{k}}}.\]
    The corresponding derivative is,
    \[S'(t)=\frac{\sqrt{k}}{(k-1)!}e^{-\frac{k-1}{\sqrt{k}}}e^{-kt+\sqrt{k}t}(kt)^{k-2}\left[(-k+\sqrt{k})kt+k(k-1)\right].\]
    The sign of $S'(t)$ is determined by $(-k+\sqrt{k})kt+k(k-1)$. Therefore, $S'(\eta)>0$ and $S'(\infty)<0$. The extreme point $t_0=(k-1)/(k-\sqrt{k})$ gives the maximum value.
    \begin{align}
    \label{eq:max_S}
        S(t_0)&=\left(\sqrt{k}\right)^{k}\frac{(\sqrt{k}+1)^{k-1}}{(k-1)!}e^{-\frac{k-1}{\sqrt{k}}-k+1}\\
        &=\sqrt{k}\left(\frac{\sqrt{k}}{\sqrt{k}-1}\right)^{k-1}\frac{(k-1)^{k-1}}{(k-1)!}e^{-(k-1)}e^{-\frac{k-1}{\sqrt{k}}}.\notag
    \end{align}
    Using the Stirling approximation $\sqrt{2\pi n}(n/e)^n\leq n!$, 
    \begin{align*}
        S(t_0)\leq \left(\frac{\sqrt{k}}{\sqrt{k}-1}\right)^{k-1}\sqrt{\frac{k}{2\pi(k-1)}}e^{-\sqrt{k}+\frac{1}{\sqrt{k}}}.
    \end{align*}
    Consider $(\sqrt{k}/(\sqrt{k}-1))^{k-1}$, substituting $r=1/\sqrt{k}$ and taking the logarithm,
    \begin{align*}
        (k-1)\ln{\left(\frac{\sqrt{k}}{\sqrt{k}-1}\right)}&=-\left(\frac{1}{r^2}-1\right)\ln(1-r)\\
        &=\left(\frac{1}{r^2}-1\right)\left(\sum_{j=1}^\infty\frac{r^j}{j}\right)\tag{Maclaurin series of $\ln(1-r)$}\\
        &=\left(\frac{1}{r}+\frac{1}{2}+\frac{r}{3}+\dots\right)-\left(r+\frac{r^2}{2}+\dots\right)\\
        &\leq \frac{1}{r}+\frac{1}{2}=\sqrt{k}+\frac{1}{2}.
    \end{align*}
    Therefore,
    \begin{align}
    \label{eq:max_S_approx}
        S(t_0)\leq e^{\frac{1}{2}+\frac{1}{\sqrt{k}}}\sqrt{\frac{k}{(k-1)2\pi}}.
    \end{align}
    This is a decreasing function of $k$. When $k\geq 9$, ~(\ref{eq:max_S_approx}) is less than 1. Meanwhile, we can verify that~(\ref{eq:max_S}) is also less than 1 for $2\leq k\leq 8$. Finally, we conclude that $R(x)\leq R(x_0)\leq S(t_0)\leq 1$ for all $k\geq 2$ and finish the proof.
    \end{proof}
\end{proof}
\
\bibliographystyle{plainurl}% the mandatory bibstyle
\bibliography{refs}

\end{document}